\documentclass[manyauthors]{fundam-arxiv}
\usepackage{graphicx}
\usepackage{enumerate}
\usepackage{url}
\usepackage{comment}
\usepackage{amsmath}
\usepackage{algorithmicx}
\usepackage{algpseudocode}
\begin{document}
\title{Finding minimum locating arrays using a CSP solver}
\author{Tatsuya Konishi\\
Graduate School of Information Science \\
and Technology, Osaka University
\and Hideharu Kojima\\
Graduate School of Information Science \\
and Technology, Osaka University
\and Hiroyuki Nakagawa\\
Graduate School of Information Science \\
and Technology, Osaka University
\and Tatsuhiro Tsuchiya\\
Graduate School of Information Science \\
and Technology, Osaka University}
\maketitle
\address{T. Tsuchiya, Graduate School of Information Science and Technology, Osaka University, Suita, Japan}

\setcounter{page}{1}
\issue{XXI~(2019)}
\runninghead{T. Konishi et al.}{Finding minimum locating arrays using a CSP solver}

\begin{abstract}
Combinatorial interaction testing is an efficient software testing
strategy. If all interactions among test
parameters or factors needed to be covered,
the size of a required test suite would be prohibitively large. In
contrast, this strategy only requires covering $t$-wise interactions
where $t$ is typically very small. As a result, it becomes possible to
significantly reduce test suite size.  Locating arrays aim to
enhance the ability of combinatorial interaction testing. In particular,
$(\overline{1}, t)$-locating arrays can not only execute all $t$-way
interactions but also identify, if any,  which of the interactions
causes a failure. In spite of this useful property,
there is only limited research either on how to generate locating arrays or on
their minimum sizes. In this paper, we propose an approach to generating
minimum locating arrays. In the approach, the problem of finding a
locating array consisting of $N$ tests is represented as a
Constraint Satisfaction Problem (CSP) instance, which is in turn
solved by a modern CSP solver.
The results of using the proposed approach reveal
many $(\overline{1}, t)$-locating arrays that are smallest known so far.
In addition, some of these arrays are proved to be minimum.
\end{abstract}


\section{Introduction}	

Combinatorial interaction testing is a testing strategy aimed to achieve high fault detection capability with a small number of tests~\cite{kuhn_introductionCT2013,nie_surveyCT2011}.
The basic form of this strategy is $t$-way testing, which requires testing all combinations of values on any $t$ factors.
The number of such combinations can be rather large; but the number of tests required to exercise these combinations is usually small enough to be executed in reasonable time
when $t$ is small, such as two or three.
In addition, it is believed that many of software defects involve a combination of a few factors~\cite{kuhn_introductionCT2013,kuhn_faultinteractions2004}.
The parameter $t$ is referred to as \emph{strength}.

For example, consider the model of a System Under Test (SUT) shown in Fig.~\ref{fig:printermodel}.
Here the SUT is a printer system.
The model consists of four factors, including layout, size, color, and duplex.
All the factors have two possible values.
Figure~\ref{fig:CA} is a test suite for the printer model, where
each row represents a test.

Such a test suite can be viewed as an $N \times k$ array where
$N$ is the number of tests and $k$ is the number of factors.
An array representing a test suite for the $t$-way testing strategy is called
a \emph{covering array} of strength $t$ or
a $t$-\emph{covering array}.
Figure~\ref{fig:CA} is a covering array of strength~2, where
all interactions of strength~2, i.e., interactions
involving two factors occur in at least one row.
Since it is impossible to cover all such interactions using only four rows,
this covering array is \emph{minimum} in that there is no smaller covering array
of the same strength.


\begin{figure}[tb]
\centering
\begin{tabular}{|c|c|c|c|} \hline
  Layout      & Size & Color & Duplex    \\ \hline \hline
  Portrait    & A4 & Yes   & On     \\
  Landscape & A5 & No    & Off \\ \hline
\end{tabular}
\caption{Model of a printer system}\label{fig:printermodel}
\end{figure}

\begin{figure}[tb]
\centering
\begin{tabular}{|c|c|c|c|c|} \hline
    & Layout   &  Size & Color & Duplex \\ \hline \hline
  1  & Portrait  & A4 & Yes & On   \\
  2  & Portrait  & A4 & No  & Off  \\
  3  & Portrait  & A5 & No  & On   \\
  4  & Landscape & A4 & No  & On   \\
  5  & Landscape & A5 & Yes & Off  \\ \hline
\end{tabular}
\caption{Covering array of strength~2}\label{fig:CA}
\end{figure}

Using a $t$-covering array as a test suite allows us to detect the presence of
a failure-triggering interaction of strength~$t$. However, even when a failure occurs,
it is not always possible to locate which interaction causes the failure.
For example, suppose that all the tests in Figure~\ref{fig:CA} were executed
and that all tests were passed except that the third test failed.
The failed test contains six two-way interactions, namely,
(Portrait, A5), (Portrait, No), (Portrait, On), (A5, No), (A5, On) and (No, On).
Three of the six interactions, namely, (Portrait, No), (Portrait, On) and (No, On), can be safely
excluded from the candidates since they occur in the passed tests.
However, it is impossible to decide
which of the remaining three interactions is failure-triggering.

\emph{Locating arrays} add to covering arrays the ability of locating
failure-triggering interactions~\cite{colbourn_locatingarray2008}.
An array is $(\overline{d}, t)$-locating if it can locate all
failure-triggering interactions
as long as the total number of these interactions is at most $d$
and their strength is $t$.

Figure~\ref{fig:LA} shows a $(\overline{1}, 2)$-locating array for the SUT model in Figure~\ref{fig:printermodel}.
Using this array enables to identify a failure-triggering interaction of strength~2, provided
that there is at most one such interaction.
For example, suppose that as a result of executing all the tests in Figure~\ref{fig:LA},
the fourth and fifth test have failed, and the other tests have passed.
In this case, it can be safely concluded that (No, On) is the failure-triggering
interaction, because this is the only interaction that occurs in these
two failed tests but not in the other ones.



\begin{figure}[tb]
\centering
\begin{tabular}{|c|c|c|c|c|} \hline
    & Layout   &  Size & Color & Duplex \\ \hline \hline
  1  & Portrait  & A4 & Yes & On  \\
  2  & Portrait  & A4 & No  & Off \\
  3  & Portrait  & A5 & Yes & Off \\
  4  & Portrait  & A5 & No  & On  \\
  5  & Landscape & A4 & No  & On  \\
  6  & Landscape & A5 & Yes & On  \\
  7  & Landscape & A5 & No  & Off \\ \hline
\end{tabular}
\caption{(${\overline 1},2)$-Locating array}\label{fig:LA}
\end{figure}

At the cost of the failure localizing ability, the size of locating arrays is usually
substantially larger than covering arrays.
For example, compare the $(\overline{1}, 2)$-locating array
shown in Figure~\ref{fig:LA}
and the covering array shown in Figure~\ref{fig:CA}.
As explained later, the locating array is minimum in size but contains four more rows.
Since the number of rows (tests) greatly affects the cost of testing,
techniques for finding small, especially minimum locating arrays are required
to use locating arrays as a practical testing tool.
Also the minium sizes of locating arrays are of theoretical interest of its own right.

So far little is known about locating arrays.
For example, minimum locating arrays are known for a few sporadic cases.
To our knowledge all previous studies on the smallest size of locating arrays
are based solely on mathematical arguments.
In contrast, we propose a computational approach to generation of
minimum locating arrays in this paper.
We formulate the problem of finding a locating array of a given size
as the Constraint Satisfaction Problem (CSP),
so that we can make full use of recent advances in CSP solving.
Locating arrays of minimum size can be obtained by repeatedly solving
the problem of varying sizes.

The rest of the paper is structured as follows.
Section~\ref{basicconcept} provides the definitions of basic concepts,
such as locating arrays and covering arrays.
Section~\ref{SAT} explains the proposed approach to finding minimum locating arrays.
Section~\ref{experiment} shows the results of using the proposed approach.
Section~\ref{relatedresearch} describes related research.
Section~\ref{conclusion} concludes this paper.

\section{Preliminaries}\label{basicconcept}

The System Under Test (SUT) is modeled as two positive integer parameters: $k$ and $v$.
The SUT has $k$ \emph{factors}, $F_1$, ..., $F_k$.
The parameter $v$ represents the number of \emph{levels} that the factors can take.
A \emph{test} is a vector of size $k$ where
the $i$th element of a test represents the level of $F_i$ in that test.
For presentation simplicity, we let the domain of each factor be $\{0, ..., v-1\}$.
We view a collection of $N$ tests as an $N \times k$ array where each row is one of the $N$ tests.
Thereafter we use the term \emph{array} to mean such an array.

An \emph{interaction} $T$ is a possibly empty subset $T$ of
$\{F_1,...,F_k\}\times\{0,...,v-1\}$
such that no two elements of $T$ share the same factor
(i.e., $\forall (F, V), (F', V')\in T:  (F, V) \neq (F', V') \Rightarrow F \neq F'$).
An interaction $T$ is $t$-way or of strength~$t$ if and only if $|T| = t$.
A row (test) covers an interaction $T$ if and only if the level on $F$ of the row
is $V$ for all $(F, V) \in T$.
Given an array $A$, we let $\rho_A(T)$ denote the set of rows that cover $T$.
Similarly, for a set $\cal T$ of interactions, we let
$\rho_A({\cal T}) = \bigcup_{T\in {\cal T}} \rho_A(T)$.
Also we let ${\cal I}_t$ denote the set of interactions of strength $t$.
For example, suppose that the SUT model consists of three factors that have two levels,
i.e., $k = 3$ and $v = 2$.
Then ${\cal I}_2$ contains, for example, $\{(F_1, 0), (F_2, 1)\}$,
$\{(F_1, 1), (F_3, 1)\}$, etc.


\begin{definition}
An array $A$ is $t$-covering if and only if
\[
\forall T \in {\cal I}_t: ~ \rho_A(T) \neq \emptyset. \label{def:CA}
\]
\end{definition}
The definition~\ref{def:CA} states that all $t$-way interactions $T$ must be
covered by some row in $A$.

Colbourn and McClary introduce several version of locating arrays~\cite{colbourn_locatingarray2008}.
Of them the following two are of our interest.
\begin{definition}
\label{def:LA1}
An array $A$ is $(d, t)$-locating if and only if
\[
\forall {\cal T}_1, {\cal T}_2 \subseteq {\cal I}_t
\mathrm{\ such\ that\ } |{\cal T}_1|=|{\cal T}_2| = d:
\rho_A({\cal T}_1) = \rho_A({\cal T}_2)
 \Leftrightarrow
{\cal T}_1 = {\cal T}_2.
\]
\end{definition}

\begin{definition}
\label{def:LA2}
An array $A$ is $(\overline{d}, t)$-locating if and only if
\[
\forall {\cal T}_1, {\cal T}_2 \subseteq {\cal I}_t
\mathrm{\ such\ that\ } |{\cal T}_1|\leq d, |{\cal T}_2| \leq d:
\rho_A({\cal T}_1) = \rho_A({\cal T}_2)
 \Leftrightarrow
{\cal T}_1 = {\cal T}_2.
\]
\end{definition}

\noindent
In these definitions,
$\rho_A({\cal T}_1) = \rho_A({\cal T}_2)
 \Leftrightarrow
{\cal T}_1 = {\cal T}_2$
is equivalent to
\[
{\cal T}_1 \neq {\cal T}_2
 \Rightarrow
\rho_A({\cal T}_1) \neq \rho_A({\cal T}_2),
\]
because
${\cal T}_1 = {\cal T}_2 \Rightarrow
\rho_A({\cal T}_1) = \rho_A({\cal T}_2)$
trivially holds for any ${\cal T}_1, {\cal T}_2$ by the definition of $\rho_A(\cdot)$
and thus what effectively matters is only the other direction (i.e.,
${\cal T}_1 = {\cal T}_2 \Leftarrow
\rho_A({\cal T}_1) = \rho_A({\cal T}_2)$).

Let us assume that an interaction is either \emph{failure-triggering}
or not and that the outcome of the execution of a test is
\emph{fail} if the test covers at least one failure-triggering
interaction; \emph{pass} otherwise.
Given the test outcome of all tests,
using $(d, t)$- and $({\overline d}, t)$-locating arrays
enables to identify all existing $t$-way fault-triggering interactions.
In the practice of testing,
we are often interested in the case where the number of fault-triggering interactions
is at most one.
For this practical reason, we restrict ourselves to investigating
$({\overline 1}, t)$-locating arrays.
However, $(d, t)$-locating arrays are still useful to characterize the properties
of $({\overline 1}, t)$-locating arrays as follows.
\begin{proposition}
An array $A$ is $({\overline 1}, t)$-locating if and only if
it is $(1, t)$-locating and $t$-covering.
\end{proposition}
\begin{proof}
Proof. (if part ($\Leftarrow$))
We show that $\rho_A({\cal T}_1) \neq \rho_A({\cal T}_2)$
for any two different set of $t$-way interactions,
${\cal T}_1$ and ${\cal T}_2$, such that $0 \leq |{\cal T}_1| \leq |{\cal T}_2|\leq 1$.
1) If $A$ is $(1, t)$-locating, then
${\cal T}_1 \neq {\cal T}_2 \Rightarrow
\rho_A({\cal T}_1) \neq \rho_A({\cal T}_2)$ holds for any
${\cal T}_1, {\cal T}_2 \subseteq {\cal I}_t$ such that
$|{\cal T}_1|=|{\cal T}_2| = 1$.
2) If $A$ is $t$-covering, then $\rho_A({\cal T}_2) \neq \emptyset$
for any ${\cal T}_2 \subseteq {\cal I}_t$  such that $|{\cal T}_2| = 1$.
If ${\cal T}_1 \subseteq {\cal I}_t$ and $|{\cal T}_1| = 0$, then
${\cal T}_1 = \emptyset$.
Since $\rho_A({\cal T}_1) = \rho_A(\emptyset) = \emptyset$,
$\rho_A({\cal T}_1) \neq \rho_A({\cal T}_2)$ for any
${\cal T}_1, {\cal T}_2 \subseteq {\cal I}_t$ such that $|{\cal T}_1|=0$ and $|{\cal T}_2| = 1$.
3) If $|{\cal T}_1|=|{\cal T}_2| = 0$, then ${\cal T}_1 = {\cal T}_2 = \emptyset$.

(only if part ($\Rightarrow$) 1) If $A$ is $({\overline 1}, t)$-locating, it is trivially
$(1, t)$-locating.
2) Let ${\cal T}' = \emptyset$.
Then $\rho_A({\cal T}') = \emptyset$.
Now let $T$ be any $t$-way interaction and ${\cal T} = \{T\}$.
If $A$ is $({\overline 1}, t)$-locating, then $\rho_A({\cal T}) \neq \rho_A({\cal T}') = \emptyset$.
Hence $\rho_A(T) \neq \emptyset$, which means that $A$ is $t$-covering.
\end{proof}

In the next section, we propose encoding schemes that represent the
problem of finding locating arrays as CSP. Following the above proposition,
the constraints in our CSP formulations represent two properties: the
property of being $(1, t)$-locating and that of being $t$-covering.

\section{Finding small locating arrays with CSP solver}\label{SAT}

In this section we address the following problem:
\begin{quote}
Given the number of factors $k$, the number of levels $v$,
the number of rows $N$, find a $({\overline 1}, t)$-locating array
for these parameter values.
\end{quote}
We propose encoding schemes that represent the problem
as a Constraint Satisfaction Problem (CSP).

In Section~\ref{basicconstraint}, we describe the basic encoding scheme which provides
the baseline CSP formulation.
In Section~\ref{alternativematrix}, we explain another encoding scheme using
additional decision variables.
In Section~\ref{symmetrybreaking}, we describe a symmetry reduction technique
where we use additional constraints to reduce the solution space of the CSP.
In Section~\ref{algorithm}, we propose an algorithm that finds small locating arrays
using the CSP encoding schemes.

\subsection{Basic constraints}\label{basicconstraint}

Here we describe a basic encoding scheme. The encoding scheme uses
integer decision variables $x_{r,i}$ ($ 1 \leq r \leq N $, $ 1 \leq i \leq k $).
The range of $x_{r,i}$ is from 0 to $v-1$.
Each $x_{r,i}$ is used to represent the value in the $r$th row of
the $i$th column of an array.

The necessary and sufficient condition that the array is $({\overline 1},t)$-locating consists of two parts, namely, the part representing
that it is $t$-covering and the part representing that it is $(1, t)$-locating.

The first part of the constraints is as follows:
\begin{equation}
\begin{array}{l}
\forall l_1, ..., l_t \in \{0, 1, ..., v-1\}, \\
\forall i_1, ..., i_t~\mathrm{such\ that}~1\leq i_1 < ... < i_t \leq k:  \\
\displaystyle \qquad \qquad \qquad
\bigvee_{1\leq r \leq N}  \bigwedge_{1\leq j \leq t} (x_{r,i_j} = l_j).
\end{array}
\label{eq:constraint1}
\end{equation}

Constraint~(\ref{eq:constraint1}) states that
every $t$-way interaction $\{(F_{i_1}, l_{1}), ...,  (F_{i_t}, l_t)\}$
occurs in at least one row represented by $r$ ($1 \leq r \leq N$).

The second part of the constraints specifies that the array is $(1, t)$-locating.
This part is formed as follows.
\begin{equation}
\begin{array}{l}
\forall l_1,...,l_t, s_1,...,s_t \in \{0, 1, ..., v-1\}, \\
\forall i_1, ..., i_t, h_1,....,h_t~\mathrm{such\ that}~1 \leq i_1 < ... < i_t \leq k,
       1 \leq h_1 < ... < h_t \leq k: \\
\displaystyle \qquad \qquad \qquad
\bigvee_{1\leq r \leq N} \big\{ \bigwedge_{1\leq j \leq t} (x_{r,i_j} = l_j)
                            \oplus
                                \bigwedge_{1\leq j \leq t} (x_{r,h_j} = s_j) \big\},
\end{array}
\label{eq:constraint2}
\end{equation}
where $\oplus $ represents exclusive-or (XOR).

Constraint~(\ref{eq:constraint2}) states that
for any two $t$-way interactions, there is at least one row, denoted by $r$,
in which only either one of them is covered.
In the formula, the two interactions are
$T_1 = \{(F_{i_1}, l_1), ...., (F_{i_t}, l_t)\}$ and
$T_2 = \{(F_{h_1}, s_1), ...., (F_{h_t}, s_t)\}$.
Let ${\cal T}_1 = \{T_1\}$ and ${\cal T}_2 = \{T_2\}$.
Clearly $\exists r: (r \in \rho_A({\cal T}_1) \land r \not \in \rho_A({\cal T}_2))
\lor
(r \not \in \rho_A({\cal T}_1) \land r \in \rho_A({\cal T}_2))$
implies $\rho_A({\cal T}_1) \neq \rho_A({\cal T}_2)$ and vice versa.
Therefore, the above constraint represents the necessary and sufficient
condition that the array is $(1, t)$-locating.

Aiming at reducing the search space, we also add the following constraint that
enforces the first row to be all zeros.
\begin{equation}
\bigwedge_{1 \leq i \leq k} (x_{1, i} = 0)
\end{equation}

\subsection{Alternative Matrix}\label{alternativematrix}

Aimed at speeding up CSP solving, here we introduce an alternative
encoding.
We call the array representation based on this encoding
the alternative matrix model, as was in \cite{hnich_constraint2006}.
In the alternative matrix model,
$\genfrac(){0pt}{1}{n}{t}$ decision variables of integer type are associated with
the $\genfrac(){0pt}{1}{n}{t}$ interactions of strength~$t$ occurring in a row
in a one-by-one manner.

To show the idea, let us consider the $11 \times 10$ array with $v = 2$
shown in Figure~\ref{fig:LA2}.
Also suppose that we are interested in the case $t = 2$.
This array is represented by the alternative matrix model as shown
in Figure~\ref{fig:alternativematrix}.
The alternative matrix has a total of $\genfrac(){0pt}{1}{10}{2} = 45$ columns
each of which corresponds to a choice of $t (=2)$ different columns of
the original array.
The alternative matrix has the same number of rows as the original array.
An entry in the alternative matrix is an integer ranging from~0
to $v^t -1$. The value represents
the $t$-way interaction that occurs on the corresponding
row and columns of the original array.
For example, the entry on the column $(1,3)$ and row~7
has value 2, meaning that $\{(F_1, 1), (F_3, 0)\}$ occurs on the 7th row
because $2 = 2^1  * 1  + 2^0 * 0$.

The encoding based on the alternative matrix model uses
a total of $N *  \genfrac(){0pt}{1}{k}{t}$ decision variables of integer type, $y_{r, (i_1,...,i_t)}$,
in addition to $x_{r,i}$ representing the array itself.
The domain of $y_{r, (i_1,...,i_t)}$ ranges from~0 to~$v^t -1$.
Each $y_{r, (i_1,...,i_t)}$ represents the interaction
represented by $x_{r,i_1},...,x_{r,i_t}$, namely,
$\{(F_{i_1}, x_{r,i_1}),...., (F_{i_t}, x_{r,i_t})\}$ in the $r$th row.
For example, the interaction $\{(F_{i_1}, x_{r,i_1}), (F_{i_2}, x_{r,i_2})\}$
of strength~2, which can be either of (0, 0), (0, 1), (1, 0), and (1, 1),
is represented by $y_{r, (i_1, i_2)}$, which takes a value 0, 1, 2, and 3.

\begin{figure}[tb]
\centerline{
\begin{tabular}{c|c}
      & 1 \ 2 \ 3 \ 4 \ 5 \ 6 \ 7 \ 8 \ 9 \ 10 \\ \hline
    1 & 0 \ 0 \ 0 \ 0 \ 0 \ 0 \ 0 \ 0 \ 0 \ 0 \\
    2 & 0 \ 0 \ 0 \ 0 \ 0 \ 1 \ 1 \ 1 \ 1 \ 1 \\
    3 & 0 \ 0 \ 1 \ 1 \ 1 \ 0 \ 0 \ 1 \ 1 \ 1 \\
    4 & 0 \ 1 \ 0 \ 1 \ 1 \ 1 \ 1 \ 0 \ 0 \ 1 \\
    5 & 0 \ 1 \ 1 \ 0 \ 1 \ 0 \ 1 \ 0 \ 1 \ 0 \\
    6 & 0 \ 1 \ 1 \ 1 \ 0 \ 1 \ 0 \ 1 \ 0 \ 0 \\
    7 & 1 \ 0 \ 0 \ 1 \ 1 \ 1 \ 0 \ 0 \ 1 \ 0 \\
    8 & 1 \ 0 \ 1 \ 0 \ 1 \ 1 \ 1 \ 1 \ 0 \ 0 \\
    9 & 1 \ 0 \ 1 \ 1 \ 0 \ 0 \ 1 \ 0 \ 0 \ 1 \\
   10 & 1 \ 1 \ 0 \ 0 \ 1 \ 0 \ 0 \ 1 \ 0 \ 1 \\
   11 & 1 \ 1 \ 0 \ 1 \ 0 \ 0 \ 1 \ 1 \ 1 \ 0 \\
\end{tabular}}
\caption{A $11\times 10$ array. (This is the minimum $({\overline 1}, 2)$-locating
array we found using the proposed approach.)}\label{fig:LA2}
\end{figure}

\begin{figure}[tb]
\centerline{
\begin{tabular}{c|cccccccccc}
 & (1,2) & (1,3) & (1,4) & (1,5) & ... & (7,10) & (8,9) & (8,10) & (9,10) \\\hline
 1 & 0 & 0 & 0 & 0 & ... & 0 & 0 & 0 & 0 \\
 2 & 0 & 0 & 0 & 0 & ... & 3 & 3 & 3 & 3 \\
 3 & 0 & 1 & 1 & 1 & ... & 1 & 3 & 3 & 3 \\
 4 & 1 & 0 & 1 & 1 & ... & 3 & 0 & 1 & 1 \\
 5 & 1 & 1 & 0 & 1 & ... & 2 & 1 & 0 & 2 \\
 6 & 1 & 1 & 1 & 0 & ... & 0 & 2 & 2 & 0 \\
 7 & 2 & 2 & 3 & 3 & ... & 0 & 1 & 0 & 2 \\
 8 & 2 & 3 & 2 & 3 & ... & 2 & 2 & 2 & 0 \\
 9 & 2 & 3 & 3 & 2 & ... & 3 & 0 & 1 & 1 \\
10 & 3 & 2 & 2 & 3 & ... & 1 & 2 & 3 & 1 \\
11 & 3 & 2 & 3 & 2 & ... & 2 & 3 & 2 & 2
\end{tabular}}
\caption{Alternative matrix for the array in Fig.~\ref{fig:LA2}}\label{fig:alternativematrix}
\end{figure}

The correspondence between the original array and the alternative
matrix can be established with the following constraints over
the decision variables.
\begin{equation}
\begin{array}{l}
\forall r \in \{1, 2,..., N\}, \\
\forall i_1,...,i_t~\mathrm{such\ that}~1\leq i_1 < i_2 <...<i_t \leq k: \\
\displaystyle \qquad \qquad
y_{r, (i_1,...,i_t)}
 = v^{t-1}*x_{r i_1} + v^{t-2} * x_{ri_2} + ... + v^{0} * x_{ri_t}
\end{array}
\label{eq:channel}
\end{equation}
(In \cite{hnich_constraint2006}, this constraint is called a channeling constraint.)

When Constraint~(\ref{eq:channel}) is imposed, it is possible to replace
Constraints~(\ref{eq:constraint1}) and~(\ref{eq:constraint2}) with
other constraints over $y_{r, (i_1,...,i_t)}$ as shown below.

Constraint~(\ref{eq:constraint1}), which specifies that the array
is $t$-covering, can now be replaced with the following one.
\begin{equation}
\begin{array}{l}
\forall L \in \{0, 1, ..., v^t-1\}, \\
\forall i_1, ..., i_t~\mathrm{such\ that}~1\leq i_1 < ... < i_t \leq k: \\
\displaystyle \qquad \qquad \bigvee_{1\leq r \leq N} (y_{r, (i_1,..., i_t)} = L)
\end{array}\label{eq:constraint1A}
\end{equation}
Constraint~(\ref{eq:constraint2}) can be replaced with:
\begin{equation}
\begin{array}{l}
\forall L, S \in \{0, 1, ..., v^t-1\}, \\
\forall i_1, ..., i_t, h_1, ..., h_t~\mathrm{such\ that}~1 \leq i_1 <...< i_t \leq k, 1 \leq h_1 <...< h_t \leq k: \\
\displaystyle \qquad \qquad
\bigvee_{1\leq r \leq N} (y_{r, (i_1, ..., i_t)} = L) \oplus (y_{r, (h_1, ...., h_t)} = S)
\end{array}
\end{equation}

Note that
$y_{r, (i_1, ..., i_t)} = L$ and $y_{r, (h_1, ..., h_t)} = S$
hold if and only if
($x_{r,{i_1}} = l_1\wedge ... \wedge x_{r,{i_t}} = l_t$)
and
($x_{r,{h_1}} = s_1\wedge ... \wedge x_{r,{h_t}} = s_t$) hold
in Constraint~(\ref{eq:constraint2}).

\subsection{Symmetry Breaking}\label{symmetrybreaking}

Here we introduce a symmetry breaking technique to reduce the solution space of the CSP
while guaranteeing the correctness of the result of the problem.
Usually a locating array has a number of symmetrically isomorphic arrays.
For example, replacing the order of rows or columns in a locating array still yields a
locating array.
The technique presented here introduces additional constraints based on
such symmetric properties.
The additional constraints prevent a solution search from exploring part of the solution space
that can be safely omitted.
We use the same technique proposed by Hnich et al\cite{hnich_constraint2006}., which
was originally used to search for covering arrays.
This technique breaks symmetry by imposing lexicographical order on rows and columns
of the array.

The lexicographical ordering of rows is specified by the following constraint.
\begin{equation}
\begin{array}{l}
\forall r \in \{1, 2, ..., N-1\}:  \\
\displaystyle \quad
\bigvee_{1 \leq i \leq k}\{\{\bigwedge_{1 \leq i' < i} (x_{r,i'} = x_{{r + 1},i'})\} \wedge
 (x_{r,i} < x_{{r + 1},i})\}
\end{array}
\label{eq:lexrow}
\end{equation}
In this constraint, given two consecutive rows
$(x_{r,1}, ... , x_{r,n})$ and $(x_{r+1,1}, ..., x_{r+1,n})$,
the first one is considered to be smaller than the second one for the lexicographical order,
if $x_{r,i}< x_{r+1,i}$ for the first $i$ where $x_{r,i}$ and $x_{r+1,i}$ differ.

Likewise, the lexicographical constraint on the columns are as follows:
\begin{equation}
\begin{array}{l}
\forall i \in \{1, 2, ..., k-1\}:  \\
\displaystyle
\{\bigvee_{1 \leq r \leq N} \{\{\bigwedge_{1 \leq r' < r} (x_{r',i} = x_{r',i+1})\}
  \wedge (x_{r,i} < x_{r,i+1})\} \} \vee \{\bigwedge_{1 \leq r \leq N} (x_{r,i} = x_{r,i+1})\}
\end{array}
\label{eq:lexcolumn}
\end{equation}
The disjunct at the right-hand side,
i.e., $\bigwedge_{1\leq r \leq N} (x_{r,i} = x_{r,i+1})$
is necessary to permit two consecutive columns to be identical.
This is contrast to the lexicographic order constraint on rows, which prohibits
consecutive rows from being the same.
A subtle point is that Constraint~(\ref{eq:lexrow}) imposes
a new condition that all rows must be different, which does not exist in the
definition of the problem defined at the beginning of Section~\ref{SAT}.
However, this does not affect the correctness of the CSP formulation
as far as $N \leq v^k$ holds.
Note that the case $N > v^k$ is of no technical interest,
since the total number of possible tests is $v^k$.


\subsection{Algorithm for finding locating arrays}\label{algorithm}

Here we presents an algorithm for searching for small, hopefully minimum, locating arrays.
The algorithm uses the CSP encodings presented above as its basis.

The outline of the algorithm is as follows. It initially sets the size $N$ of the array, i.e.,
the number of rows, to the known lower bound on the size of the minimum $(\overline{1}, t)$-locating
array.
Then the value of $N$ is repeatedly incremented until a locating array is found.
The problem of deciding if a locating array of size $N$ exists is solved
by having a CSP solver solve the CSP instance encoded using the proposed techniques.

It should be noted that the result of the CSP solving has three possibilities.
If a \emph{satisfying} valuation, that is, a value assignment to the decision variables that satisfies the constraints is found, then the existence of
a $(\overline{1}, t)$-locating array of size $N$ is proved, because
the assignment represents such a locating array.
If the solver proves that the CSP instance is unsatisfiable, i.e.,
no satisfying valuations exist, then it is possible
to conclude the nonexistence of a locating array of that size.
The remaining possibility is that the solver is not able to decide whether
the CSP instance is satisfiable or not.
This could happen because of memory shortage or time out.
Below shows the whole algorithm.


\begin{description}
\item[Given] strength $t$, the number of factors $k$, the number of values $v$
\item[Step~1.] Set
  $N \gets$ known lower bound on the minimum size
  and $isMinimum \gets \mathrm{true}$.  \
\item[Step~2]  Solve the CSP instance encoded with $(N, t, k, v)$.
\item[Step~3-1]  Case 1. [Satisfiable.]
Output the satisfying valuation and the current values of $N$ and $isMinimum$.
Terminate the algorithm.
\item[Step~3-2]  Case 2. [Unsatisfiable.]
  Set $N \gets N + 1$ and go to Step~2.
\item[Step~3-3]  Case 3. [No answer is obtained.]
  Set $isMinimum \gets \mathrm{false}$ and $N \gets N + 1$. Go to Step~2.
\end{description}

The algorithm outputs a satisfying valuation, which represents a $({\overline 1}, t)$-locating
array. The locating array is the smallest one that the algorithm finds in a single run.
The value of $N$ output by the algorithm represents the size of the array obtained.
Note that the locating array may not be guaranteed to be minimum, since the CSP solving
may have failed to solve a CSP instance with some smaller $N$ (Case~3).
Boolean variable $isMinimum$ is set to false if this case happens and the guarantee
is lost. Therefore, the obtained locating array is minimum
only when the output value of $isMinimum$ is true.

In Step~1, the initial value of $N$ is set to a known lower bound on the size of
the minimum $({\overline 1}, t)$-locating array.
The bound could be $v^t$, which trivially holds because all $t$-way interactions
must occur on any $t$ factors.
A tighter bound was proved by Tang et al.~\cite{tang_optimalitylocatingarray_2012}
and therefore we use it for this step.

\section{Experiments}\label{experiment}

This section describes the experiments we conducted.
The experiments consists of two parts.
The first part is intended to answer the question of which encoding scheme works best.
In the second part, the proposed method is applied using the best encoding scheme.
In the experiments, we focus on the case $t=2$, because
\emph{pair-wise testing}, that is, testing of interactions of strength~2,
is the most practiced form of combinatorial interaction testing~\cite{aetg,pict}.

We implemented the proposed method mainly using Scarab~\cite{soh_scarab2013}, which is a tool
that provides a domain-specific language based on Scala programming
language for describing CSP instances. Scarab also embeds Sugar~\cite{Tamura2009} together
with SAT4J as its default CSP solver. Sugar translates a CSP
instance to an instance of the Boolean satisfiability problem (SAT).
SAT4J is a SAT solver~\cite{sat4j}. Sugar and SAT4J are both written in Java. In
addition to the default setting of Scarab, we also used a customized
version of it which can output the intermediate SAT instance. This
customization allowed us to use a faster SAT solver than the default one.


\subsection{Comparison between different encoding schemes}
\label{subsec:exp1}

In the first experiment, we compare a total of four encoding schemes as follows.
\begin{itemize}
\item Na\"{i}ve
\item Na\"{i}ve + Symmetry breaking
\item Alternative
\item Alternative + Symmetry breaking
\end{itemize}
Using these encoding schemes, we ran the proposed method within the
parameter ranges as follows:
$3\leq k \leq 12, v = 2$ and $3\leq k \leq 8, v = 3$.
The timeout was set to one hour.
It should be noted that even for the same values of $N, k, v$,
it was the case that a locating array was obtained
with some encoding scheme but a timeout occurred with another scheme.
The increment of $N$ was stopped when a locating array had been obtained
using at least one encoding scheme.
As a result, a total of 34 CSP instances were tested for each encoding scheme.
\begin{figure}[t]
\centerline{\includegraphics[width=0.75\textwidth]{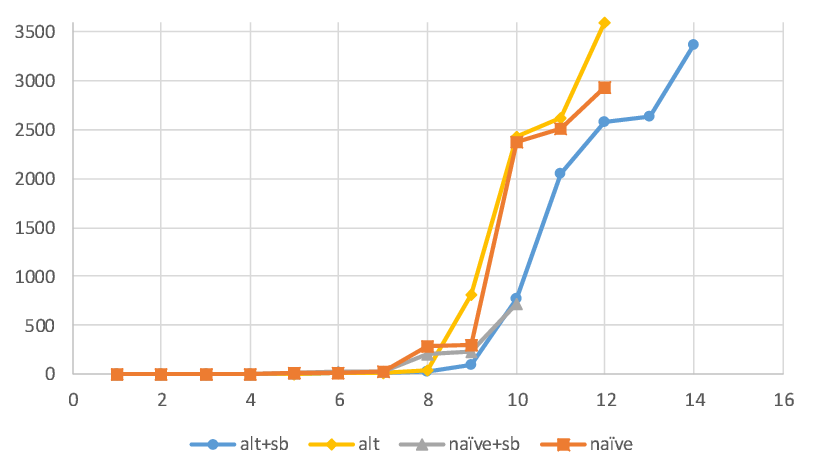}}
\caption{Comparison of the four encodings.
The horizontal axis shows the number of CSP instances solved within
one hour, while the vertical axis shows running time.
The encoding
based on the alternative matrix model and symmetry breaking (alt+sb),
 which showed best performance, solved 14 out of 34 CSP instances.}
\label{fig:encodings}
\end{figure}

We carried out this experiment on a Mac High Sierra 10.13.3 Laptop equipped with
a Core i5 CPU (1.8~GHz) and 8 GByte memory.
Figure~\ref{fig:encodings} shows the results of the experiments
in the form of cactus plot.
The curves show the runtime needed to solve a certain number of problem instances.
That is, a point $(x, y)$ on the curve represents that
$x$ problem instances required at most $y$ seconds to be solved, whereas the remaining
$34-x$ instances could not be solved within $y$ seconds.


As can be seen from the figure, the encoding based on
the alternative matrix model with symmetry breaking
exhibited the best performance among the four different encoding schemes.
When using this encoding scheme, 14 out of the 34 CSP instances were solved within
the one hour limit.
The other encoding schemes solved less instances with longer running times.
All the instances solved by these encodings turned out to be satisfiable.

\subsection{Finding minimum locating arrays}\label{experiment:SAT}

\begin{table}[t]
\caption{Results for various $N$ ($v = 2$).
S: Satisfiable, U: Unsatisfiable, T: Timeout (12 hours)\label{tbl:v2}}
{\small
\begin{tabular}{lllllllllllllllll} \hline
       & 6 & 7 & 8 & 9 & 10 & 11 & 12 & 13 & 14 & 15 & 16 & 17 & 18 & 19 \\ \hline
$2^3$  & S &   &   &   &    &    &    &    &    &    &    &    &    &            \\
$2^4$  &   & S &   &   &    &    &    &    &    &    &    &    &    &            \\
$2^5$  &   &   & S &   &    &    &    &    &    &    &    &    &    &            \\
$2^6$  &   &   &   & S &    &    &    &    &    &    &    &    &    &            \\
$2^7$  &   &   &   &   & S  &    &    &    &    &    &    &    &    &            \\
$2^8$  &   &   &   &   & U  & S  &    &    &    &    &    &    &    &           \\
$2^9$  &   &   &   &   &    & S  &    &    &    &    &    &    &    &            \\
$2^{10}$ &   &   &   &   &    & S  &    &    &    &    &    &    &    &    \\
$2^{11}$ &   &   &   &   &    & S  &    &    &    &    &    &    &    &    \\
$2^{12}$ &   &   &   &   &    & U  & S  &    &    &    &    &    &    &    \\
$2^{13}$ &   &   &   &   &    &    & T  & T  & S  &    &    &    &    &    \\
$2^{14}$ &   &   &   &   &    &    & T  & T  & T  & S  &    &    &    &    \\
$2^{15}$ &   &   &   &   &    &    & T  & T  & T  & S  &    &    &    &    \\
$2^{16}$ &   &   &   &   &    &    & T  & T  & T  & S  &    &    &    &    \\
$2^{17}$ &   &   &   &   &    &    & T  & T  & T  & T  & S  &    &    &    \\
$2^{18}$ &   &   &   &   &    &    & T  & T  & T  & T  & S  &    &    &    \\
$2^{19}$ &   &   &   &   &    &    & T  & T  & T  & T  & T  & S  &    &    \\
$2^{20}$ &   &   &   &   &    &    & T  & T  & T  & T  & T  & S  &    &    \\
$2^{21}$ &   &   &   &   &    &    & T  & T  & T  & T  & T  & T  & S  &    \\
$2^{22}$ &   &   &   &   &    &    & T  & T  & T  & T  & T  & T  & S  &    \\
$2^{23}$ &   &   &   &   &    &    & T  & T  & T  & T  & T  & T  & T  & S  \\
\hline
\end{tabular}
}
\end{table}

\begin{table}[t]
\caption{Results for various $N$ ($v = 3$).
S: Satisfiable, U: Unsatisfiable, T: Timeout (12 hours)\label{tbl:v3}}
{\small
\begin{tabular}{lllllllllllllll} \hline
   & 14 & 15 & 16 & 17 & 18 & 19 & 20 & 21 & 22 & 23 & 24 & 25 & 26 & 27 \\ \hline
$3^{3}$ & U  & S  &    &    &    &    &    &    &    &    &    &    &    &    \\
$3^{4}$ &    &    & S  &    &    &    &    &    &    &    &    &    &    &    \\
$3^{5}$ &    &    &    & S  &    &    &    &    &    &    &    &    &    &    \\
$3^{6}$ &    &    &    & S  &    &    &    &    &    &    &    &    &    &    \\
$3^{7}$ &    &    &    &    & T  & T  & T  & T  & T  & S  &    &    &    &    \\
$3^{8}$ &    &    &    &    &    &    & T  & T  & T  & T  & T  & S  &    &    \\
$3^{9}$ &    &    &    &    &    &    &    & T  & T  & T  & T  & T  & T  & S  \\ \hline
\\ \hline
         & 23 & 24 & 25 & 26 & 27 & 28 & 29 & 30 & 31 & 32 & 33 & 34 & 35 & 36  \\
\hline
$3^{10}$ & T  & T  & T  & T  & T  & T  & S  &    &    &    &    &    &    &   \\
$3^{11}$ & T  & T  & T  & T  & T  & T  & T  & T  & T  & S  &    &    &    &   \\
$3^{12}$ &    & T  & T  & T  & T  & T  & T  & T  & T  & T  & T  & S  &    &   \\
$3^{13}$ &    &    & T  & T  & T  & T  & T  & T  & T  & T  & T  & T  & T  & S \\
\hline
\end{tabular}}


\end{table}

\begin{table}[t]
\caption{Sizes of $({\overline 1}, 2)$-locating arrays obtained.
The column \textit{size} represents the number of rows, 
\textit{time} represents the time taken by Glucose to find the
satisfying valuation that represents the locating array, 
$\geq$ shows the known lower bound on the size~\cite{tang_optimalitylocatingarray_2012}, and 
\textit{minimum?} indicates whether the array
is the minimum or not guaranteed to be so at present.
\label{tbl:summary}}
{\small
\begin{tabular}{lrrrc} \hline
         & size & time    & $\geq$         & minimum? \\ \hline
$2^{3}$  & 6    & 0.02    & 6             & Yes     \\
$2^{4}$  & 7    & 0.07    & 7             & Yes     \\
$2^{5}$  & 8    & 0.2     & 8             & Yes       \\
$2^{6}$  & 9    & 0.6     & 9             & Yes       \\
$2^{7}$  & 10   & 1.3     & 10            & Yes       \\
$2^{8}$  & 11   & 2       & 10            & Yes     \\
$2^{9}$  & 11   & 122.5   & 11            & Yes       \\
$2^{10}$ & 11   & 476.6   & 11            & Yes       \\
$2^{11}$ & 11   & 709.9   & 11            & Yes       \\
$2^{12}$ & 12   & 6977.7  & 11            & Yes     \\
$2^{13}$ & 14   & 859.8   & 12            & ?       \\
$2^{14}$ & 15   & 586.2   & 12            & ?       \\
$2^{15}$ & 15   & 1835.6  & 12            & ?       \\
$2^{16}$ & 15   & 29793.7 & 12            & ?       \\
$2^{17}$ & 16   & 10543.8 & 12            & ?       \\
$2^{18}$ & 16   & 15509.7 & 12            & ?       \\
$2^{19}$ & 17   & 10312.7 & 12            & ?       \\
$2^{20}$ & 17   & 29147.9 & 12            & ?       \\
$2^{21}$ & 18   & 9699    & 12            & ?       \\
$2^{22}$ & 18   & 8869.2  & 12            & ?       \\
$2^{23}$ & 19   & 13557.9 & 12            & ?      \\ \hline
\end{tabular}\quad \quad
\begin{tabular}{lrrrc} \hline
         & size & time    & $\geq$         & minimum? \\ \hline
$3^{3}$  & 15   & 0.1     & 14            & Yes     \\
$3^{4}$  & 16   & 1       & 16            & Yes       \\
$3^{5}$  & 17   & 2503.8  & 17            & Yes       \\
$3^{6}$  & 17   & 270     & 17            & Yes       \\
$3^{7}$  & 23   & 3918.7  & 18            & ?       \\
$3^{8}$  & 25   & 10581.3 & 20            & ?       \\
$3^{9}$  & 27   & 3622.8  & 21            & ?       \\
$3^{10}$ & 28   & 24627.1 & 22            & ?       \\
$3^{11}$ & 31   & 1579    & 22            & ?       \\
$3^{12}$ & 33   & 9479.4  & 23            & ?       \\
$3^{13}$ & 35   & 15229.3 & 24            & ?      \\ \hline
\\
\\
\\
\\
\\
\\
\\
\\
\\
\\
\end{tabular}
}
\end{table}

The results shown in Section~\ref{subsec:exp1} suggest that the alternative matrix model
with the symmetry breaking technique works best among the four encoding schemes.
Using the best encoding we searched for small $(\overline{1}, 2)$-locating arrays within a larger range of the parameters.
This experiment was conducted on a faster PC equipped with a 2.40GHz Intel Xeon E5-2665 CPU
and 128 Gbyte memory. The OS was CentOS~7.2.
To find smaller arrays, we changed the way of solving CSP instances as follows.
In the first experiment, we had used SAT4J, which is embedded in Scarab, to solve the SAT instances representing the CSP instances.
In the second experiment, we first produced SAT instances using a customized version of Scarab.
We then solve the SAT instances using Glucose~\cite{audemard_predicting2009}, a faster SAT solver.
Glucose can make an explicit use of parallelism in a multicore CPU.
In our case, we ran the solver with 16 threads.
We also enlarged the timeout period from one hour to 12 hours.
The timeout was enforced on the execution of Glucose; no timeout was used for the execution of
Scarab.

Tables~\ref{tbl:v2} and~\ref{tbl:v3} present how the proposed
algorithm ran for each problem instance. The leftmost column shows the
problem instances in the form of $v^k$, where $k$ is the number of
factors and $v$ is the number of levels for each factor. Each of the
remaining columns corresponds to a value of $N$ which is the size of an
array represented by a CSP instance. The value of $N$ is indicated in
the topmost row. A non-empty entry in the table has either of S, U
or T, which stands for Satisfiable, Unsatisfiable, or Timeout,
respectively. The execution of the proposed algorithm
exhibited three different patterns.
The first pattern applies to small instances including
 $2^{3}$, ..., $2^{7}$, $2^{9}$, ..., $2^{11}$,
$3^{4}$, ..., $3^{6}$.
In this case, the CSP instance was turned to be satisfiable
when $N$ is the lower bound by Tang et al~\cite{tang_optimalitylocatingarray_2012}.
Hence the locating arrays found are minimum.
The second pattern applies to $2^{8}$, $2^{12}$ and $3^{3}$, where
the CSP instances were unsatisfiable when $N$ was the lower bound
and satisfiable when $N$ was increased by one.
The arrays obtained for these problem instances are also minimum,
because the nonexistence of smaller arrays was proved.
For the remaining problem instances, the SAT solver had timed out for
all $N$ until a locating array was found. Therefore it is not possible
to conclude that the obtained locating arrays are minimum for now.
To our knowledge, however, the existence of all these arrays has not been known,
except the one for $2^{3}$, which was already shown in
\cite{tang_optimalitylocatingarray_2012}.

Table~\ref{tbl:summary} summarizes the size of
the locating arrays obtained.
The column \textit{size} represents the size, i.e., the number of rows of
the $({\overline 1}, 2)$-locating array found for each problem.
The column \textit{time} represents the time taken by Glucose to find the
satisfying valuation that represents the locating array.
The column $\geq$ shows the known lower bound on the size of the minimum locating array~\cite{tang_optimalitylocatingarray_2012}.
The column \textit{minimum?} indicates whether the obtained array
is the minimum or not guaranteed to be so at present.

\section{Related works}\label{relatedresearch}

The notion of locating arrays was proposed by Colbourn and McClary in~\cite{colbourn_locatingarray2008}.
Concrete applications of locating arrays are discussed
in~\cite{Aldaco:2015,colbourn_locatingarray2008,7562157,7528941}.
At present there is not much work on locating arrays.
Mathematical properties of locating arrays are studied in~\cite{Spread2016,Shi2012,tang_optimalitylocatingarray_2012}.
The minimum sizes of $(1, 1)$-, $({\overline 1}, 1)$-, $(1, {\overline 1})$-,
and $({\overline 1}, {\overline 1})$-locating
arrays are proved in \cite{Spread2016}.
The minimum $(\overline{2}, t)$-locating arrays for a few special cases
are provided in \cite{Shi2012}.
A lower bound on the minimum size of $(1, t)$-locating arrays is provided
in~\cite{tang_optimalitylocatingarray_2012}.
The optimality of some $(1, t)$-locating arrays that have few factors
(three to five) is also proved in~\cite{tang_optimalitylocatingarray_2012}.
Three ``cut-and-paste'' type constructions for
$(1, \overline{2}$)-locating arrays are proposed in~\cite{colbourn2016}.
Studies on computational constructions include 
\cite{konishi_LASAT2017,nagamoto_locatingpairwisetesting2014,10.1007/978-3-319-94667-2_29,Colbourn2018}. 
The two-page paper~\cite{konishi_LASAT2017} reported early results
of the work presented in this paper.
In \cite{nagamoto_locatingpairwisetesting2014} we proposed a 
``one-test-at-a-time'' greedy heuristic
algorithm for constructing $(1, 2)$-locating arrays.
The use of \textit{resampling}~\cite{Moser:2010:CPG:1667053.1667060} for locating 
array construction is studied in \cite{10.1007/978-3-319-94667-2_29,Colbourn2018}. 
The greedy heuristic algorithm and resampling techniques 
cannot be used to obtain lower bounds on the sizes of locating arrays. 
In~\cite{2018arXiv180106041J} we formulate locating arrays in the presence of constraints are formulated as \textit{constrained locating arrays}. 
Computational constructions for this variant of locating arrays 
can be found in~\cite{Jin2018,8639696}.


In contrast, covering arrays have been extensively studied.
Solvers for the Boolean satisfiability problem (SAT) are often used to construct
covering array. Research in this line
 includes~\cite{hnich_constraint2006,Banbara:2010,nanba_satsolving2012,DBLP:conf/icst/YamadaKACOB15}.
There are a number of other techniques for constructing
covering arrays, including algebraic constructions~\cite{Colbourn2005}, greedy heuristic algorithms~\cite{pict,boraz_ACTS2012,cohen_AETG1997},
metaheuristic search algorithms~\cite{walker_tabusearch2009,986326,Shiba:2004}.
Unlike these methods, the SAT-based methods have a useful feature that
they can be used to prove the non-existence of covering arrays of a given size.
The approach proposed in this paper can be viewed as an adaptation of
these SAT-based covering array constructions to locating arrays.

\section{Conclusion}\label{conclusion}

In this paper, we proposed an approach to finding small $({\overline 1},
t)$-locating arrays. We represented the problem of finding a locating
array of size $N$ as an instance of the Constraint Satisfaction Problem
(CSP). The proposed approach repeatedly solves this problem using a CSP
solver, gradually increasing $N$. We were successful in obtaining
minimum $({\overline 1}, t)$-locating arrays for some cases where the
number of factors is small. Also for other cases, we were able to find
locating arrays that are the smallest known so far.


\end{document}